\documentclass[reqno]{amsart}
\usepackage{color}
\usepackage{eucal}
\usepackage{accents}

\newtheorem{lemma}{Lemma}[section]
\newtheorem{proposition}{Proposition}[section]

\theoremstyle{definition}
\newtheorem{definition}{Definition}[section]
\theoremstyle{remark}
\newtheorem{remark}{Remark}[section]

\newcommand{\ad}{\underline a}

\newcommand{\bd}{\underline b}
\newcommand{\uu}{\overline u}
\newcommand{\ud}{\underline u}
\newcommand{\vu}{\overline v}
\newcommand{\vd}{\underline v}
\newcommand{\hu}{\overline H}
\newcommand{\hd}{\underline H}

\newcommand{\gu}{\overline G}

\newcommand{\eu}{\overline\eta}
\newcommand{\ed}{\underline\eta}
\newcommand{\cu}{\overline c}
\newcommand{\cd}{\underline c}
\newcommand{\du}{\overline d}

\newcommand{\ru}{\overline r}
\newcommand{\rd}{\underline r}

\numberwithin{equation}{section}

\begin{document}
\title{Direct \lq\lq Delay\rq\rq\ Reductions of the Toda Equation}
\author{Nalini Joshi}\thanks{The research reported in this paper was supported by the Australian Research Council Discovery Program grant \#DP0559019.}
\address{School of Mathematics and Statistics F07, The University of Sydney, NSW 2006 Australia}
\email{nalini@maths.usyd.edu.au}
\date{28 October 2008}
\begin{abstract}
A new direct method of obtaining reductions of the Toda equation is described. We find a canonical and complete class of all possible reductions under certain assumptions. The resulting equations are ordinary differential-difference equations, sometimes referred to as delay-differential equations. The representative equation of this class is hypothesized to be a new version of one of the classical Painlev\'e equations. The Lax pair associated to this equation is obtained, also by reduction.
\end{abstract}
\maketitle
\section{Introduction}
Many papers have been written on the Toda equation
\begin{align}
\label{eq:toda}&\ \left\{\begin{array}{ll}
                                u_t&=u\,(\vu - v)\\
                                v_{t}&=2\,(u^2-\ud^2)
                                \end{array}
                                \right.
\end{align}
where $u=u_n(t)=u(n,t) $, $v=v_n(t)=v(n,t)$, $u_t=\partial u(n, t)/\partial t$, $\uu=u(n+1, t)$, $\ud=u(n-1,t)$, $\vu=v(n+1, t)$, $\vd=v(n-1,t)$. The familiar exponential form of the Toda equation
\begin{align}\label{eq:exp toda}
\frac{d^2\,Q_n}{dt^2}=\exp\bigl\{-(Q_n-Q_{n-1})\bigr\}-\exp\bigl\{-(Q_{n+1}-Q_{n})\bigr\}
\end{align}
can be obtained through the change of variables
\begin{align}\label{eq:transf toda}
u_n&=\frac12\,\exp\bigl\{-(Q_n-Q_{n-1})/2\bigr\}\\
v_n&=-\,\frac12\,\frac{dQ_{n-1}}{dt}
\end{align}
In either case, the solutions depend on two independent variables $(n, t)$. Equations that involve iterates in one independent variable and derivatives  in another are referred to as differential-difference equations. Below, we also use this term for equations in which iterates and derivatives in the {\em same} variable appear. 

The Toda equation is also an example of a \lq\lq lattice model\rq\rq. Such lattice models appear in many physical settings, ranging from the study of thermalization in metals to the study of cellular neural networks and optical lattices. For these applications, it is essential to understand the wide variety of solutions possible. Reductions provide one method of extending our knowledge of the space of known solutions.

We focus on reductions to equations involving only one independent variable
\begin{equation}\label{eq:dde1}
H_\eta=K(H, \hu, \hd), \quad H:{\mathbb R}\ \mapsto\ {\mathbb R}
\end{equation}
where $H_\eta=dH/d\eta$, $\hu=H(\eu, t)$, $\hd=H(\ed,t)$. Examples of such reductions of the Toda equation were obtained through Lie symmetry analysis by Levi and Winternitz \cite{lw:06}. 

\subsection{Background} Symmetry reductions of differential-difference equations have been studied by many authors. A comprehensive review can be found in \cite{lw:06}. Early results on symmetry reductions kept $n$ fixed while allowing $u$ and $t$ to be deformed by a method based on the classical approach for differential equations that was developed by Sophus Lie. An interesting variation on this standard approach was considered in \cite{qcs:92} by allowing $n$ to deform continuously in addition to $u$ and $t$. The resulting equations arising as reductions are ordinary differential-difference equations of the form (\ref{eq:dde1}). 

In \cite{qcs:92}, Quispel {\it et al} found an equation of the form (\ref{eq:dde1}) as a reduction of the Kac-van Moerbeke or Volterra equation. They showed that their reduced equation  becomes the classical first Painlev\'e equation \cite{ince} in a continuum limit. This equation was described as a \lq\lq delay-differential\rq\rq \footnote{We note that this reduced equation contains both retarded $H(\eta-1)$ and advanced terms $H(\eta+1)$, whilst the usual terminology limits the usage of the term \lq\lq delay\rq\rq\ to equations containing only retarded terms such as $H(\eta-1)$. } version of a Painlev\'e equation by Grammaticos {\em et al.} \cite{grm:93}. Other such equations were proposed by Grammaticos {\em et al.} \cite{grm:93} as delay Painlev\'e equations by using the criterion of the singularity confinement method. In \cite{lw:06}, other ordinary differential-difference equations were found as reductions of the Toda equation. We show that these are contained in our results. Moreover, we show that our results are complete under the assumptions given below.

\subsection{Direct Method} In 1989, Clarkson and Kruskal \cite{ck:89} used a  \lq\lq direct approach\rq\rq to find reductions of the Boussinesq equation and found new reductions which were not captured by the classical Lie symmetry approach.  This direct approach was later shown to be related to \lq\lq non-classical\rq\rq symmetries. We develop a direct approach to finding reductions of differential-difference equations. 

The most general form for a reduction is
\begin{equation}\label{eq:gen ansatz}
u(n, t)=U(n, t, H(\eta), G(\eta)),\ v(n, t)=V(n, t, H(\eta), G(\eta)),\quad \eta=\eta(n, t),
\end{equation}
where $H$ and $G$ form a coupled system of equations of the form (\ref{eq:dde1}).
For Equations (\ref{eq:toda}), it turns out to be sufficient to take the ansatz
\begin{subequations}\label{eq:ansatz}
\begin{align}
u(n, t)&=a(n, t) + b(n, t)\, H(\eta)\\
v(n, t)&=c(n,t) + d(n,t)\, G(\eta)
\end{align}
\end{subequations}
where $\eta=\eta(n, t)$. Central to our argument are the following rules (stated for $a$, $b$ and $H$ for conciseness, but they apply also to $c$, $d$ and $G$)
\begin{enumerate}
\item[Rule 1:]\label{rule1} If $a(n,t)=a_0(n, t) + b(n, t)\,\Gamma(\eta)$, then we can take $\Gamma\equiv 0$ w.l.o.g. by substituting $H(\eta)\mapsto H(\eta)-\Gamma(\eta)$.
\item[Rule 2:]\label{rule2} if $b(n,t)$ has the form $b(n,t)=b_0(n,t)\,\Gamma(\eta)$, then we can take $\Gamma\equiv 1$ w.l.o.g. by substituting $H(\eta)\mapsto H(\eta)/\Gamma(\eta)$.
\item[Rule 3:]\label{rule3} If $\eta(n,t)$ is determined by an equation of the form $\Gamma(\eta)=\eta_0(n,t)$, where $\Gamma$ is invertible, then we can take $\Gamma(\eta)=\eta$ w.l.o.g. by substituting $\eta\mapsto \Gamma^{-1}(\eta)$.
\end{enumerate}
\begin{definition}\label{defn: reduction transf}
Given non-zero, differentiable and invertible functions, $\Gamma(\eta)$, we refer to the transformations $H(\eta)\mapsto H(\eta)-\Gamma(\eta)$, $H(\eta)\mapsto H(\eta)/\Gamma(\eta)$, and $\eta\mapsto \Gamma^{-1}(\eta)$ as the {\em reduction transformations} on (\ref{eq:ansatz}).
\end{definition}
Our main result is given in Proposition \ref{prop:reduction} below.
\subsection{Outline of Results} 
In this paper, we obtain direct reductions of Equations (\ref{eq:toda}). The details of our direct method are given in Section 2. In Section 3, we show that the ans\"atze (\ref{eq:ansatz}) in fact represent the general case and can be assumed without loss of generality. A continuum limit of the resulting differential-difference equations is given in Section 4. In Section 5, we find corresponding reduction of the Lax pair for the Toda equation. Finally, a conclusion rounds off the paper.

\section{Direct Reduction of the Toda Equation}
\begin{proposition}\label{prop:reduction}
Suppose that the ans\"atze (\ref{eq:ansatz}) hold. Then the only possible non-linear second-order reduction of Equation (\ref{eq:toda}) of the form (\ref{eq:dde1}) that is unique up to reduction transformations of $H$, $G$ and $\eta$ is given by 
\begin{subequations}\label{eq:reduced toda}
\begin{align}\label{eq:reduced toda 1}
-c_0\,H+H_\eta&=H\bigl(\gu -G\bigr)\\
\label{eq:reduced toda 2}
p_0-c_0\,G+G_\eta&=2\,\bigl(H^2-\hd^2\bigr)
\end{align}
\end{subequations}
where the reduction is given by $\eta(n, t)=\nu(n)+\sigma(t)$, $\nu(n)$ being an arbitrary function, with 
\begin{equation}
\sigma(t)=\left\{\begin{array}{ll}
			\frac{1}{c_0}\,\log(c_0\,t+c_1)+c_2&{\rm if\ }c_0\not=0\\
			a_0\,t+a_1&{\rm otherwise}\\
			\end{array}\right.
\end{equation}
where $c_j$, $j=0,\ldots, 2$, $a_0$, $a_1$ are constants and the reductions of $u$, $v$ are given by the following two respective cases
\begin{enumerate}
\item Case $c_0\not=0$: 
\begin{subequations}
\begin{align}
\label{h trans 1}
u(n, t)&=\pm\,\frac{1}{c_0\,t+c_1}\,H(\eta)\\
\label{g trans 1}
v(n, t)&=-\,\frac{p_0}{c_0}\,\frac{1}{(c_0t+c_1)}+c_3+\frac{1}{c_0\,t+c_1}\,G(\eta)
\end{align}
\end{subequations}
\item Case $c_0=0$:
\begin{subequations}
\begin{align}
\label{h trans 2}
u(n, t)&=\pm\,a_0\,H(\eta)\\
\label{g trans 2}
v(n, t)&=p_0\,a_0^2\,t+p_1+a_0\,G(\eta)
\end{align}
\end{subequations}
with $p_0$, $p_1$ being constants.
\end{enumerate}
\end{proposition}
\begin{proof}
Under the ans\"atze (\ref{eq:ansatz}), the Toda equation becomes
\begin{subequations}
\begin{align}\label{eq: toda 1}
a_t+b_t\,H+b\,\eta_t\,H_\eta&=(a+b\,H)\bigl(\cu-c+\du\,\gu-d\,G\bigr)\\
\label{eq: toda 2}
c_t+d_t\,G+d\,\eta_t\,G_\eta&=2\bigl(a^2-{\ad}^2+2\,a\,b\,H-2\ad\,\bd\hd+b^2\,H^2-\bd^2\hd^2\bigr)
\end{align}
\end{subequations}
In the following we indicate generic functions of $\eta$ (which are assumed to be differentiable, non-zero and invertible) by the notation $\Gamma_j(\eta)$. 

Since we seek non-linear reduced equations of the form (\ref{eq:dde1}), we require the terms $H\gu$ and $HG$ to be present in the reduced equation (\ref{eq: toda 1}). Therefore, we require that 
\begin{subequations}
\begin{eqnarray}\label{eq:b1}
b\eta_t\,\Gamma_1(\eta)&=&b\,d\\
\label{eq:d1}d\,\Gamma_2(\eta)&=&\du
\end{eqnarray}
\end{subequations}
Consider Equation (\ref{eq:b1}), which implies $d=\eta_t\,\Gamma_1(\eta)$. However, by Rule 2, this implies $\Gamma_1(\eta)\equiv 1$ and, therefore, $d=\eta_t$ w.l.o.g. 

Using this in Equation (\ref{eq:d1}), we obtain $\overline{\eta_t}=\eta_t\Gamma_2(\eta)$. However, by taking a change of variables
\[
\eta=\Omega\bigl(\xi(n,t)\bigr),\ {\rm where\ }\ \overline\Omega_\xi=\Omega_\xi\,\Gamma_2\bigl(\Omega(\xi)\bigr),
\]
we can take $\overline{\eta_t}=\eta_t$, w.l.o.g. This implies that $\eta_t=\sigma'(t)$, for some differentiable function $\sigma(t)$, and therefore, $\eta(n,t)=\nu(n)+\sigma(t)$, where $\nu(n)$ is an arbitrary function.

Now consider the second equation (\ref{eq: toda 2}). Requiring that the terms $H^2$ and $G_\eta$ both remain in the reduced equation, we find that we must have
\[
b^2=d\eta_t\,\Gamma_3(\eta)\ \Rightarrow\ b^2=(\sigma'(t))^2,\ {\rm w.l.o.g.}
\]
where we have used Rule 2 once again. Therefore, we have $b=\pm\,\sigma'(t)$.

Now consider the linear terms in Equation (\ref{eq: toda 1}). In particular, if we require that terms linear in $\gu$ remain in the reduced equation, we must have
\[
a\,\du=b\eta_t\Gamma_4(\eta)\ \Rightarrow\ a=\pm\sigma'(t)\Gamma_4(\eta)=b\,\Gamma_4(\eta)\ \Rightarrow\ a\equiv 0, \ {\rm w.l.o.g.}
\]
by an application of Rule 1. 

If, on the other hand, the linear term in $H$ on the left side of the equation remains in the reduced equation, then
\[
b_t=bd\,\Gamma_5(\eta)\ \Rightarrow\ \sigma''=(\sigma')^2\,\Gamma_5(\eta).
\]
However, since $\sigma$ only depends on $t$, while $\eta$ also depends on $n$, this can only hold if $\Gamma_5$ is identically constant. Letting this constant be $-c_0$, we find $\sigma''(t)=-c_0\,(\sigma'(t))^2$. We integrate this ODE for $\sigma(t)$ to find
\[
\sigma(t)=\left\{\begin{array}{ll}
			\frac{1}{c_0}\,\log(c_0\,t+c_1)+c_2&{\rm if\ }c_0\not=0\\
			a_0\,t+a_1&{\rm if\ }c_0=0\\
			\end{array}\right.
\]
where $c_j$, $j=0,1, 2$ and $a_0$, $a_1$ are arbitrary constants.

Finally, if the linear term in $H$ on the right side of Equation (\ref{eq: toda 1}) remains, then
\[
b\,(\cu -c)=b\,\eta_t\,(\overline\Gamma_6-\Gamma_6)
\]
where we have taken the liberty of writing the arbitrary function of $\eta$ on the right as the difference of another such function. This leads to
\[
\cu-c=\sigma'\,(\overline\Gamma_6-\Gamma_6)\ \Rightarrow\ c=\sigma'\Gamma_6(\eta)+\gamma(t),
\]
where $\gamma(t)$ is an arbitrary function of $t$. However, since $\sigma'=d$, we now have $c=d\Gamma_6(\eta)+\gamma(t)$. By Rule 2, we can take $\Gamma_6\equiv 0$ and, therefore, $c=\gamma(t)$ w.l.o.g.

The reduced equations (\ref{eq: toda 1}-\ref{eq: toda 2}) are now
\begin{align*}
-c_0\,H+H_\eta&=H\,\bigl(\gu-G\bigr)\\
\frac{\gamma'(t)}{(\sigma'(t))^2}-c_0\,G+G_\eta&=2\bigl(H^2-\hd^2\bigr)
\end{align*}

Since this equation can only contain coefficients that are functions of $\eta$, we must have $\gamma'=p_0\,(\sigma')^2$, where $p_0$ is a constant. That is, 
\[
\gamma=\left\{\begin{array}{ll}
			-\,\frac{p_0}{c_0}\,\frac{1}{(c_0t+c_1)}+c_3,&{\rm if\ }c_0\not=0\\
			p_0\,a_0^2\,t+p_1&{\rm if\ }c_0=0\\
			\end{array}\right.
\]
where $c_3$ and $p_1$ are constants. 
\end{proof}
\begin{remark}
Our reduced equations (\ref{eq:reduced toda}) form a system of differential difference equations, which evolves on a sequence of domains containing points 
$${\mathcal P}=\bigl\{\eta_0,\, \overline{\eta_0},\, \overline{\overline{\eta_0}}, \, \ldots\bigr\},$$
where if $\eta_0=\nu(n)+\sigma(t)$, then $\overline{\eta_0}=\nu(n+1)+\sigma(t)$. In the interior of any domain in $n$, where the mapping $\nu(n)\mapsto \nu(n+1)$ is defined, we get a semi-infinite chain of points ${\mathcal P}$ and a corresponding sequence of domains on which these iterates are defined. Since $\nu(n)$ is an arbitrary function, we have an infinite-dimensional family of reductions. 
\end{remark}
\begin{remark}
We note that two cases of non-linear reductions were found by Levi and Winternitz \cite{lw:06}.  In these cases, the new independent variable, called $y$ in their paper, is given respectively by (A) $y=t\,\exp(-\alpha\,n)$ or (B) $y=t-\alpha\,n$, where $\alpha$ is a constant. Case (A), after taking $\log y$ as a new variable is the sub-case $c_0=1$, $c_1=0$, $c_2=0$ of the above result. Case (B), is the sub-case $c_0=0$, $a_0=1$, $a_1=0$ of the above result. In both cases, $\nu(n)=-\alpha\,n$.
\end{remark}

\section{Generalization of Ans\"atze}
Here we show how the ans\"atze (\ref{eq:ansatz}) represent the general case. 
Consider the general reduction
\begin{subequations}
\begin{align}\label{eq:general ansatz 1}
u(n, t)&=c\bigl(n, t, H(\eta(n,t)), G(\eta(n,t))\bigr)\\
\label{eq:general ansatz 2} v(n, t)&=d\bigl(n, t, H(\eta(n,t)), G(\eta(n,t))\bigr)
\end{align}
\end{subequations}
Under these transformations, the Toda equation becomes
\begin{subequations}
\begin{align}\label{eq: general toda 1}
c_t+c_H\,H_\eta\,\eta_t+c_G\,G_\eta\,\eta_t&=c\,(\du -d)\\
\label{eq: general toda 2} d_t+d_H\,H_\eta\,\eta_t+d_G\,G_\eta\,\eta_t&=2\,(c^2-\cd^2)
\end{align}
\end{subequations}
For the reduced equation to each contain the non-linear terms in $H$, $G$, we require
\begin{subequations}
\begin{align}\label{step 1}
\du&=d\,\Gamma_1(\eta, H, G)\\
\label{step 2}
c^2&=\cd^2\,\Gamma_2(\eta, H, G)
\end{align}
\end{subequations}
Rewriting $\Gamma_1$ and $\Gamma_2$ in these equations appropriately, we can sum each to get
\[
c=e(t)\,\Gamma_3(\eta, H, G),\ d=f(t)\,\Gamma_4(\eta, H, G).
\]
Therefore, 
\begin{align*}
c_H\eta_t&=e(t)\,\Gamma_4(\eta, H, G)\,\eta_t=e(t)\,\Gamma_3(\eta, H, G)\,f(t)\Gamma_4(\eta, H, G)\\
&\Rightarrow\ c_H=g(n, t)\, \Gamma_5(\eta, H, G)\\
& \Rightarrow\ c=h(n, t)\, \Gamma_6(\eta, H, G)+k(n, t).
\end{align*}
Similarly, we find
\[
d=r(n, t)\,\Gamma_7(\eta, H, G)+s(n, t).
\]
By defining $\Gamma_6$ and $\Gamma_7$ to be new variables, if necessary, we reproduce the ans\"atze assumed earlier as (\ref{eq:ansatz}).
\section{Continuum Limits}
To find continuum limits, it is useful to first convert the system of equations (\ref{eq:reduced toda}) to a single scalar equation. 
\begin{lemma}
Define $r(\eta)=2\,\log\bigl(2\,H(\eta)\bigr)$. Then $r(\eta)$ satisfies
\begin{equation}\label{eq: r}
r_{\eta\eta}=c_0\bigl(r_\eta-2\,c_0\bigr)+\,\bigl(\exp(\ru)-2\,\exp(r)+\exp(\rd)\bigr)
\end{equation}
\end{lemma}
\begin{proof}
The proof is by direct calculation. 
\begin{align*}
r(\eta)=2\,\log(2\,H)\, &
\Rightarrow\ r_\eta=2\,\frac{H_\eta}{H}=2\left(c_0+\gu-G\right)\\
&\Rightarrow\ r_{\eta\eta}=2\,c_0\bigl(\gu-G\bigr)+4\,\bigl({\hu}^2-H^2\bigr)-2\,\bigl(H^2-{\hd}^2\bigr)
\end{align*}
which gives the desired result upon using $H^2=\exp(r)/4$.
\end{proof}
Assume for simplicity that $\overline\eta=\nu_0+\eta$. Then a continuum limit of Equation (\ref{eq: r}) is obtained by assuming 
\begin{align*}
r=\epsilon^2\,w(z),\ 
z=\epsilon\,\eta,\ 
\nu_0^2=1,\ 
c_0=\epsilon^4\,K
\end{align*}
as $\epsilon\to0$. Then $w(z)$ satisfies 
\begin{equation}
w_{zz}+6\,w^2+\alpha\,z+\beta=0
\end{equation}
where $\alpha$ and $\beta$ are integration constants. Note that this equation is a scaled and translated version of the classical first Painlev\'e equation. 

\section{Lax pair}
We apply the reduction found in \S 2 to the Lax pair for the Toda equation given by Wadati and Toda \cite{wt:75}:
\begin{subequations}
\begin{align}\label{lax 1}
&u\,\overline\psi+\ud\,\underline\psi+v\,\psi=\lambda\,\psi\\
\label{lax 2}&\psi_t=u\overline\psi-\ud\,\underline\psi
\end{align}
\end{subequations}
Using the results of Proposition \ref{prop:reduction} for the case $c_0\not=0$, we find from Equation (\ref{lax 1})
\begin{align}
\pm\,\sigma'\,H\,\overline\psi+\pm\,\sigma'\,\hd\,\underline\psi+\biggl\{\Bigl(-\,\frac{p_0}{c_0}\,\sigma'+c_3\Bigr)+\sigma'\,G\biggr\}\,\psi=\lambda\,\psi
\end{align}
In order for this to be a reduced equation containing coefficients that are only functions of $\eta$, we require $c_3=0$ and need to choose a definite sign in $b=\pm\sigma'$, say the positive sign. Furthermore, we define a new spectral variable $\zeta=\lambda/\sigma'$. Then we obtain for $\phi(\eta, \zeta)=\psi(n, t)$,
\[
\psi_t=\zeta_t\,\phi_\zeta +\sigma'\,\phi_\eta = \sigma'\,(c_0\,\zeta\phi_\zeta+\phi_\eta).
\]
Hence we get the reduced Lax pair
\begin{subequations}\label{reduced lax}
\begin{align}\label{reduced lax 1}
&H\,\overline\phi+\hd\,\underline\phi+\Bigl(G-\,\frac{p_0}{c_0}\Bigr)\,\phi=\zeta\,\phi\\
\label{reduced lax 2}&c_0\,\zeta\phi_\zeta+\phi_\eta=H\overline\phi-\hd\,\underline\phi
\end{align}
\end{subequations}
We note that the character of the Lax pair has changed from a spectral problem to a monodromy problem, since now derivatives in $\zeta$ also appear in the linear problem. By differentiating Equation (\ref{reduced lax 1}) in two different ways, once with respect to $\zeta$ and once with respect to $\eta$, and using (\ref{reduced lax 2}) to replace $\phi_\zeta$, while using (\ref{reduced lax 1}) to replace $\overline\phi$, we can show that the compatibility conditions for the linear system (\ref{reduced lax}) are precisely Equations (\ref{eq:reduced toda}). 
\begin{remark}
We note that the classical first Painlev\'e equation has no explicit solutions expressible in terms of previously known functions. There is no reason to believe that Equations (\ref{eq:reduced toda}) should possess any such explicit solutions. However, the reduced equations (\ref{eq:reduced toda}) do have the solution $H\equiv 0$, $G\equiv 0$, if $p_0=0$. The corresponding solution of the Lax pair (\ref{reduced lax}) is $\phi\equiv 0$. These results underly the solvability of the reduced equation and the reduced Lax pair. Local existence of solutions  to respective initial value problems can be inferred from the standard existence theorems for differential-difference equations.
\end{remark}

\section{Concluding Remarks}
We have developed a new direct method of obtaining reductions of differential-difference equations and applied this method to the Toda equation (\ref{eq:toda}).The completeness of our results in the class of \lq\lq reduction transformations\rq\rq\  (see Definition \ref{defn: reduction transf}) is shown in two steps. First,  in Section 2, we found the most general possible reduction subject to the ans\"atze (\ref{eq:ansatz}). Second, in Section 3, we showed that the ans\"atze assumed in fact represent the general case. A continuum limit of the reduced equations is found to be related to the first Painlev\'e equation in Section 4. Finally, in Section 5, we showed that the reduced equations inherit a linear problem from the Toda equation. 

The reductions we have found through the direct method appear to be closely related to those obtained by the Lie-symmetry approach \cite{lw:06}. However, we note that in the latter approach, the iterations of the independent variable $\eta$ are assumed to be the simple ones described by $\overline\eta=\eta+\alpha$, or $\overline\eta=\alpha\,\eta$, where $\alpha$ is a constant. No such assumption is made in our approach. The reduced equations (\ref{eq:reduced toda}) allow a more general functional iteration of $\eta$ to appear. We also note that neither a continuum limit of these reductions or their linear problem has been found before. 

The inheritance of a linear problem through the reduction indicates that Equations (\ref{eq:reduced toda}) are integrable. These results suggest that the system (\ref{eq:reduced toda}) is analogous to the well known reductions of completely integrable partial differential equations, namely the classical Painlev\'e equations. However, questions remain open on how close such an analogy might be.

The solutions of the Painlev\'e equations have a characteristic complex analytic structure in the complex plane of their independent variable. Formal Laurent expansions of the solutions of the Toda equation and one of its reductions has been discussed in \cite{rgt:93}. However, while Equations (\ref{eq:reduced toda}) form a first-degree, second-order system of differential equations in $\eta$, they form a second-degree, second-order system when considered as a system of difference equations. This introduces a multi-valuedness into the system when it is solved for the highest iterate of $H$, which makes any proof of meromorphicity of the solutions slightly delicate. Moreover, the simultaneous presence of iterates of $\phi$ and derivatives of $\phi$ in $\eta$ makes the linear system (\ref{reduced lax}) a type of monodromy problem that has not been studied before.

\end{document}